\documentclass[proceedings]{stacs}
\stacsheading{2009}{409--420}{Freiburg}
\firstpageno{409}

\usepackage{amssymb}
\usepackage{color}
\usepackage{enumerate}
\usepackage{subfig}
\usepackage[a4paper]{hyperref}

\newtheorem{mytheorem}{Theorem}
\newtheorem{mylemma}{Lemma}
\newtheorem{mycorollary}{Corollary}

\newlength{\originalitemsep}
\newenvironment{pseudocode} 
  {\begin{enumerate}[1]%
   \setlength{\itemindent}{1.5ex}%
   \setlength{\originalitemsep}{\itemsep}\setlength{\itemsep}{-\parsep}\normalsize \sf}
  {\setlength{\itemsep}{\originalitemsep}\end{enumerate}}
\newlength{\pseudocodeindent}
\newcommand{\pcindent}[1]{\rule{#1\pseudocodeindent}{0mm}}

\newcommand{\pccomment}[1]{\hfill $\triangleright$ \textsl{#1}}

\newcommand{\bdd}{\textrm{bdd-$d$-set}}

\newcommand{\cs}{\textrm{commitment structure}}
\newcommand{\BDD}{\textrm{\sc Bounded-Degree Deletion}}

\newcommand{\J}{J}

\newcommand{\findextremal}{\textbf{find\_extremal}}
\newcommand{\computeCD}{\textbf{compute\_AB}}

\newcommand{\alg}[4]{%
  \hrule
  \vspace{1ex}
  \normalsize
  \textbf{Algorithm:} #1\\
  \textbf{Input:} #2\\
  \textbf{Output:} #3
  #4
  \hrule
}

\newcommand{\proc}[4]{%
  \hrule
  \vspace{1ex}
  \normalsize
  \textbf{Procedure:} #1\\
  \textbf{Input:} #2\\
  \textbf{Output:} #3
  #4
  \hrule
}

\begin{document}
\title[A Generalization of Nemhauser and Trotter's Local Optimization
Theorem]{A Generalization of Nemhauser and Trotter's\\ Local Optimization
Theorem}
\author[pcru]{M.R. Fellows}{Michael R. Fellows}
\address[pcru]{%
 PC Research Unit, Office of DVC (Research),
 \newline University of Newcastle, Callaghan, NSW~2308, Australia.}
\email{michael.fellows@newcastle.edu.au}
\author[jena]{J. Guo}{Jiong Guo}
\author[jena]{H. Moser}{Hannes Moser}
\author[jena]{R. Niedermeier}{Rolf Niedermeier}
\address[jena]{%
Institut f\"ur Informatik, Friedrich-Schiller-Universit\"at Jena,
\newline Ernst-Abbe-Platz 2, D-07743 Jena, Germany.}
\email{{(guo,moser,niedermr)@minet.uni-jena.de}}
\thanks{%
  The first author was supported by the Australian Research Council. Work
  done while staying in Jena as a recipient of the Humboldt
  Research Award of the Alexander von Humboldt Foundation,
  Bonn, Germany.
  The second author was supported by the DFG, Emmy Noether research group PIAF, NI~369/4, and project DARE, GU~1023/1.
  The third author was supported by the DFG, projects ITKO, NI~369/5, and AREG, NI~369/9.}
%
%
\begin{abstract}
The Nemhauser-Trotter local optimization theorem applies to the
NP-hard \textsc{Vertex Cover} problem and has applications 
in approximation as well as parameterized algorithmics.
We present a framework that generalizes Nemhauser and Trotter's result
to vertex deletion and graph packing problems, introducing novel
algorithmic strategies based on purely combinatorial arguments 
(not referring to linear programming as the Nemhauser-Trotter result
originally did).

We exhibit our framework using a generalization of \textsc{Vertex Cover},
called 
\BDD{}, that has promise to become an important tool in the analysis 
of gene and other biological networks.
For some fixed~$d\geq 0$, \BDD{} asks to delete as few
vertices as possible from a graph in order to transform it into a graph
with maximum vertex degree at most~$d$.  \textsc{Vertex Cover} is the special case of
$d=0$.  Our generalization of the Nemhauser-Trotter theorem
implies that \BDD{} has a problem kernel with a
linear number of vertices
for every constant~$d$.  
We also outline an application
of our extremal combinatorial approach to the problem of packing stars with a bounded number
of leaves.
Finally, charting the border between (parameterized) tractability and
intractability for \BDD{}, we provide a W[2]-hardness result for 
\BDD{} in case of unbounded $d$-values.
%
\end{abstract}

\maketitle

\section{Introduction}
Nemhauser and Trotter~\cite{NT75} proved a famous theorem in combinatorial optimization.
In terms of the NP-hard \textsc{Vertex Cover}\footnote{\textsc{Vertex
Cover} is the following problem:
Given an undirected graph, find a minimum-cardinality set~$S$ of vertices such that
each edge has at least one endpoint in~$S$.}
problem, it can be formulated as follows:

\vspace{1em}\noindent \textbf{NT-Theorem~\cite{NT75,BE85}}.
\emph{
  For an undirected graph~$G=(V,E)$ one can compute in polynomial time two
  disjoint vertex subsets~$A$ and~$B$, such that the following three properties hold:
  \begin{enumerate}
    \setlength{\itemsep}{-\parsep}
    \item If~$S'$ is a vertex cover of the induced subgraph~$G[V \setminus (A \cup B)]$,
          then~$A \cup S'$ is a vertex cover of~$G$.
    \item There is a minimum-cardinality vertex cover~$S$ of~$G$ with~$A
          \subseteq S$.
    \item Every vertex cover of the induced subgraph~$G[V \setminus (A \cup B)]$ has 
          size at least~$|V \setminus (A \cup B)|/2$.
  \end{enumerate}}
In other words, the NT-Theorem provides a polynomial-time data reduction
for \textsc{Vertex Cover}.  That is, for vertices in~$A$ it can already be
decided in polynomial time to put them into the solution set and vertices
in~$B$ can be ignored for finding a solution. 
The NT-Theorem is very useful for approximating
\textsc{Vertex Cover}. The point is that the search for an approximate
solution can be restricted to the induced subgraph~$G[V \setminus (A
\cup B)]$. The NT-Theorem directly delivers a factor-$2$
approximation for \textsc{Vertex Cover} by choosing~$V \setminus B$ as the
vertex cover. Chen et al.~\cite{CKJ01} first observed that the NT-Theorem
directly yields a $2k$-vertex problem kernel for \textsc{Vertex Cover},
where the parameter~$k$ denotes the size of the solution set. Indeed,
this is in a sense an ``ultimate'' kernelization result in parameterized
complexity analysis~\cite{DF99,FG06,Nie06} because there is good reason
to believe that there is a matching lower bound~$2k$ for the kernel
size unless~P$=$NP~\cite{KR08}.

Since its publication numerous authors have referred to the
importance of the NT-Theorem from the viewpoint of polynomial-time
approximation algorithms (e.g.,~\cite{BE85,Khu02})
as well as from the viewpoint of parameterized algorithmics 
(e.g.,~\cite{AFLS07,CKJ01,CC08}). 
The relevance of
the NT-Theorem comes from both its practical usefulness in solving
the {\sc Vertex Cover} problem as well as its
theoretical depth having led to numerous further studies and follow-up
work~\cite{AFLS07,BE85,CC08}.
In this work, our main contribution is to provide a more general and
more widely
applicable version of the NT-Theorem.  The corresponding
algorithmic strategies and proof techniques, however, are not achieved
by a generalization of known proofs of the NT-Theorem but are completely
different and are based on extremal combinatorial arguments.
{\sc Vertex Cover} can be formulated as the problem 
of finding 
a
minimum-cardinality set of vertices whose deletion makes a graph
edge-free, that is, the remaining vertices have degree~$0$.  Our main
result is to prove a generalization of the NT-Theorem that helps in
finding a minimum-cardinality set of vertices whose deletion leaves a
graph of maximum degree~$d$ for arbitrary but fixed~$d$.  Clearly, $d=0$ is the
special case of {\sc Vertex Cover}.

\paragraph{\textbf{Motivation.}}
Since the NP-hard \BDD{} problem---given a graph and two positive
integers~$k$ and~$d$, find at most~$k$ vertices whose deletion leaves
a graph of maximum vertex degree~$d$---stands in the center of our
considerations, some more explanations about its relevance follow.
\BDD{} (or its dual problem) already appears in some theoretical
work, e.g.,~\cite{BF01,KHMN07,NRT05}, but so far it has received
considerably less attention than \textsc{Vertex Cover}, one of the best
studied problems in combinatorial optimization~\cite{Khu02}.  To advocate
and justify more research on \BDD{}, we describe an application
in computational biology. In the analysis of genetic networks based
on micro-array data, recently a clique-centric approach has shown great
success~\cite{BCK+05,CLS+05}.  Roughly speaking, finding
cliques or near-cliques (called paracliques~\cite{CLS+05}) has been
a central tool.  Since finding cliques is computationally hard (also with
respect to approximation), Chesler et al.~\cite[page~241]{CLS+05}
state that ``cliques are identified through a transformation to
the complementary dual \textsc{Vertex Cover} problem and the use
of highly parallel algorithms based on the notion of fixed-parameter
tractability.'' More specifically, in these \textsc{Vertex Cover}-based
algorithms polynomial-time data reduction (such as the NT-Theorem) plays
a decisive role~\cite{Lan08pers} (also see~\cite{AFLS07}) for efficient
solvability of the given real-world data. However, since biological
and other real-world data typically contain errors, the demand for
finding cliques (that is, fully connected subgraphs) often seems
overly restrictive and somewhat relaxed notations of cliques are more
appropriate. For instance, Chesler et al.~\cite{CLS+05} introduced
paracliques, which are achieved by greedily extending the found cliques
by vertices that are connected to almost all (para)clique vertices.
An elegant mathematical concept of ``relaxed cliques'' is that
of $s$-plexes%
\footnote{Introduced in~1978 by Seidman and Foster~\cite{SF78} in the
context of social network analysis. Recently, this
concept has again found increased interest~\cite{BBHS06,KHMN07}.}
where one demands that each $s$-plex vertex does not need to be
connected to all other vertices in the $s$-plex but to all but~$s-1$.
Thus, cliques are $1$-plexes.  The corresponding problem to find
maximum-cardinality $s$-plexes in a graph is basically as computationally
hard as clique detection is~\cite{BBHS06,KHMN07}.
However, as \textsc{Vertex Cover} is the dual problem for clique
detection, \BDD{} is the dual problem for $s$-plex detection: An
$n$-vertex graph has an $s$-plex of size~$k$ iff its complement
graph has a solution set for \BDD{} with~$d=s-1$ of size~$n-k$, and
the solution sets can directly be computed from each other.  
The \textsc{Vertex
Cover} polynomial-time data reduction algorithm has played 
an important role in the practical success
story of analyzing real-world genetic and other biological networks~\cite{BCK+05,CLS+05}. 
Our new polynomial-time data reduction algorithms for \BDD{} have the potential
to play a similar role.

\paragraph{\textbf{Our results.}} 
Our main theorem can be formulated as follows.

\vspace{1ex}\noindent \textbf{BDD-DR-Theorem (\autoref{thm:bddkernelmain})}.
\emph{
For an undirected $n$-vertex and $m$-edge graph~$G=(V,E)$, we can compute two disjoint 
vertex subsets~$A$ and~$B$ in~$O(n^{5/2} \cdot m + n^3)$ time, such that the following three 
properties hold: 
\begin{enumerate}
    \setlength{\itemsep}{-\parsep}
\item If~$S'$ is a solution set for \BDD{} of the induced subgraph~$G[V \setminus (A \cup B)]$, then~$S:=S'\cup A$ is a solution set
for \BDD{} of~$G$. 
\item There is a minimum-cardinality solution set~$S$ for \BDD{} of~$G$ with~$A \subseteq S$.
\item Every solution set for \BDD{} of the induced subgraph~$G[V \setminus (A \cup B)]$ has size at least
      $$\frac{|V\setminus (A \cup B)|}{d^3 + 4d^2+6d+4}
      .$$ 
\end{enumerate} 
}
In terms of parameterized algorithmics, this gives
a~$(d^3 + 4d^2 + 6d + 4) \cdot k$-vertex problem kernel for \BDD{}, which is
linear in~$k$ for constant $d$-values, thus joining a number of
other recent ``linear kernelization results''~\cite{BP08,Guo08,GN07ICALP,KPXS09}.
%
Our general result specializes to a $4k$-vertex problem kernel for
{\sc Vertex Cover} (the NT-Theorem provides a size-$2k$ problem kernel),
but applies to a larger class of problems.  For instance, a slightly
modified version of the BDD-DR-Theorem (with essentially the same proof)
yields a $15k$-vertex problem kernel for the problem of packing at
least~$k$ vertex-disjoint length-$2$ paths of an input graph, giving
the same bound as shown in work focussing on this problem~\cite{PS06}.%
\footnote{Very recently, Wang et
al.~\cite{WNFC08} improved the $15k$-bound to a $7k$-bound.
We claim that our kernelization based on the BDD-DR-Theorem method can
be easily adapted to also deliver the $7k$-bound.}
For the problem, where, given an undirected graph, one seeks a set of
at least~$k$ vertex-disjoint stars\footnote{A \emph{star} is a tree where all of the 
vertices but one are leaves.} of the same constant size, 
we show that a kernel with a linear
number of vertices can
be achieved, improving the best previous quadratic kernelization~\cite{PS06}.
We emphasize that our data reduction technique is based on extremal
combinatorial arguments; the resulting combinatorial kernelization
algorithm has practical potential and implementation work is underway.
Note that for~$d=0$ our algorithm computes the same type of structure
as in the ``crown decomposition'' kernelization for \textsc{Vertex Cover} 
(see, for example,~\cite{AFLS07}). However, for~$d \geq 1$ the structure returned
by our algorithm is much more complicated;
in particular, unlike for {\sc Vertex Cover} crown decompositions, in the BDD-DR-Theorem the set~$A$ is not necessarily a
separator and the set~$B$ does not necessarily form an independent set.

Exploring the borders of parameterized tractability of \BDD{}
for arbitrary values of the degree value~$d$, we show the following.
\begin{mytheorem}\label{thm:w2hard}
  For unbounded~$d$ (given as part of the input), \BDD{} is~$W[2]$-complete
  with respect to the parameter~$k$ denoting 
  the number of vertices to delete.
\end{mytheorem}
In other words, there is no hope for fixed-parameter tractability with
respect to the parameter~$k$ in the case of unbounded $d$-values.
Due to the lack of space the proof of \autoref{thm:w2hard} and several
proofs of lemmas needed to show \autoref{thm:bddkernelmain} are omitted.

\section{Preliminaries}

A \emph{\bdd} for a graph~$G=(V,E)$ is a vertex subset whose
removal from~$G$ yields a graph in which each vertex has degree at
most~$d$. The central problem of this paper is 
\begin{quote}
  \BDD{}
  \vspace{-\parsep} 
  \begin{description}
    \setlength{\itemsep}{-\parsep}
    \item[Input] An undirected graph~$G=(V,E)$, and integers~$d \geq 0$ and~$k > 0$.
    \item[Question] Does there exist a \bdd~$S \subseteq V$
                     of size at most~$k$ for~$G$?
  \end{description}
\end{quote}

In this paper, for a graph~$G=(V,E)$ and a vertex set~$S \subseteq
V$, let~$G[S]$ be the subgraph of~$G$ induced by~$S$ and~$G - S :=
G[V \setminus S]$. The \emph{open neighborhood} of a vertex~$v$ or a
vertex set~$S \subseteq V$ in a graph~$G=(V,E)$ is denoted as~$N_G(v) := \{ u \in V \mid
\{u,v\} \in E\}$ and~$N_G(S) := \bigcup_{v \in S} N_G(v) \setminus S$,
respectively. The \emph{closed neighborhood} is denoted as~$N_G[v] :=
N_G(v) \cup \{v\}$ and~$N_G[S] := N_G(S) \cup S$.
We write~$V(G)$ and~$E(G)$ to denote the vertex and edge set of~$G$,
respectively.
A \emph{packing}~$P$ of a graph~$G$ is a set of pairwise vertex-disjoint
subgraphs of~$G$.
A graph has maximum degree~$d$
when every vertex in the graph has degree at most~$d$.
A graph property is called {\em hereditary} if every induced subgraph
of a graph with this property has the property as well.

Parameterized algorithmics~\cite{DF99,FG06,Nie06} is an approach to
finding optimal solutions for NP-hard problems. 
A common method in parameterized algorithmics is to
provide polynomial-time executable \emph{data reduction rules} that lead to a \emph{problem kernel}~\cite{GN07SIGACT}.
This is the most important concept for this paper.
Given a parameterized problem instance~$(I,k)$, a data reduction rule
replaces~$(I,k)$ by an instance~$(I',k')$ in polynomial time such
that~$|I'| \leq |I|$,~$k' \leq k$, and~$(I,k)$ is a Yes-instance if and
only if~$(I',k')$ is a Yes-instance.
A parameterized problem is said to have a problem kernel, or,
equivalently, \emph{kernelization}, if, after
the exhaustive application of the data reduction rules, the resulting
reduced instance has size~$f(k)$ for a function~$f$ depending only on~$k$.
Roughly speaking, the kernel size~$f(k)$ plays a similar role in the subject of problem kernelization as the
approximation factor plays for approximation algorithms.


\section{A Local Optimization Algorithm for Bounded-Degree Deletion}\label{sec:kernel}
The main result of this section is the following generalization of the 
Nemhauser-Trotter-Theorem~\cite{NT75} for \BDD{} 
with constant~$d$.

\begin{mytheorem}[BDD-DR-Theorem]\label{thm:bddkernelmain}
For an $n$-vertex and $m$-edge graph~$G=(V,E)$, we can compute two disjoint 
vertex subsets~$A$ and~$B$ in~$O(n^{5/2} \cdot m + n^3)$ time, such that the following three 
properties hold: 
\begin{enumerate}
    \setlength{\itemsep}{-\parsep}
\item If~$S'$ is a \bdd{} of~$G-(A \cup B)$, then~$S:=S'\cup A$ is a \bdd{}
of~$G$. 
\item There is a minimum-cardinality \bdd~$S$ of~$G$ with~$A \subseteq S$.
\item Every \bdd{} of~$G-(A \cup B)$ has size at least
      $\frac{|V\setminus (A \cup B)|}{d^3 + 4d^2+6d+4}$.
\end{enumerate} 
\end{mytheorem} 
This first two properties are called the {\em local optimality conditions}. 
The remainder of this section is dedicated to the proof of this theorem. 
More specifically, we present an algorithm called \computeCD{} (see \autoref{fig:computeCD})
which outputs two sets~$A$ and~$B$ fulfilling the three properties given in~\autoref{thm:bddkernelmain}. 
The core of this algorithm is the procedure \findextremal{} (see~\autoref{fig:bdd-compression}) 
running in~$O(n^{3/2} \cdot m + n^2)$ time. This procedure returns two disjoint vertex subsets~$C$ 
and~$D$ that, among others, satisfy the local optimality conditions. The procedure is 
iteratively called by \computeCD{}. The overall output sets~$A$ and~$B$ then are the union of the 
outputs of all applications of \findextremal{}. 
Actually, \findextremal{} 
searches for~$C\subseteq V$,~$D\subseteq V$,~$C\cap D=\emptyset$ satisfying the following two 
conditions: 
\begin{enumerate}[\bf C1]
    \setlength{\itemsep}{-\parsep}
\item \label{comstr:D} Each vertex in~$N_G[D] \setminus C$ has degree
                            at most~$d$ in~$G - C$, and
\item \label{comstr:bddset} $C$ is a minimum-cardinality bdd-$d$-set for~$G[C
                                \cup D]$.
\end{enumerate}
It is not hard to see that these two conditions are stronger than 
the local optimality conditions of 
\autoref{thm:bddkernelmain}:

\begin{mylemma}\label{prop:commitmentreduction}
  Let~$C$ and~$D$ be two vertex subsets satisfying conditions~C\ref{comstr:D}
and~C\ref{comstr:bddset}. Then, the following is true: 
\begin{enumerate}[(1)]
    \setlength{\itemsep}{-\parsep}
\item If~$S'$ is a \bdd{} of~$G-(C\cup D)$, then~$S:=S'\cup C$ is a \bdd{}
of~$G$. 
\item There is a minimum-cardinality \bdd~$S$ of~$G$ with~$C\subseteq S$. 
\end{enumerate}
\end{mylemma}
\autoref{prop:commitmentreduction} will be used in the proof 
of~\autoref{thm:bddkernelmain}---it helps to make the description of the 
underlying algorithm and the corresponding correctness proofs more accessible. 
As a direct application of~\autoref{thm:bddkernelmain}, 
we get the following corollary. 
\begin{mycorollary} 
\BDD{} with constant~$d$ admits a problem kernel with at most~$(d^3 + 4d^2+6d+4)\cdot k$ vertices, 
which is computable in~$O(n^{5/2} \cdot m + n^3)$ time.
\end{mycorollary}
We use the following easy-to-verify forbidden subgraph characterization of bounded-degree graphs:
A graph~$G$ has maximum degree~$d$ if
and only if there is no ``$(d+1)$-star'' in~$G$. 
\begin{definition}
  For~$s \geq 1$, the graph~$K_{1,s} = (\{u,v_1,\ldots,v_s\},
  \{\{u,v_1\},\ldots,\{u,v_s\}\})$ is called an \emph{$s$-star}.
  The vertex~$u$ is called the \emph{center} of the star.
  The vertices~$v_1,\ldots,v_s$ are the \emph{leaves} of the star.
  A \emph{${\leq}s$-star} is an $s'$-star with~$s' \leq s$. 
\end{definition}
Due to this forbidden subgraph characterization of bounded-degree graphs, 
we can also derive a linear kernelization for the {\sc $(d+1)$-Star Packing} problem.
In this problem, given an undirected
graph, one seeks for at least~$k$ vertex-disjoint $(d+1)$-stars
for a constant~$d$.
With a slight modification of the proof of~\autoref{thm:bddkernelmain},
we get the following corollary.

\begin{mycorollary} \label{cor:star-packing}
{\sc $(d+1)$-Star Packing} admits a problem kernel with at
most~$(d^3 + 4d^2+6d+4)\cdot k$ vertices, which is computable in~$O(n^{5/2}
\cdot m + n^3)$ time.
\end{mycorollary} 
For~$d \geq 2$, the best known kernelization result was a~$O(k^2)$ kernel~\cite{PS06}.
Note that the special case of {\sc $(d+1)$-Star Packing} with~$d=1$
is also called {\sc $P_3$-Packing}, a problem well-studied in the
literature, see~\cite{PS06,WNFC08}.
\autoref{cor:star-packing} gives a 
$15k$-vertex problem kernel.
The best-known bound is~$7k$~\cite{WNFC08}.
However, the improvement from the formerly best bound~$15k$~\cite{PS06} is
achieved by improving a properly defined witness structure by local modifications.
This trick also works with our approach, that is, 
we can show that the NT-like approach also yields a~$7k$-vertex problem kernel for
\textsc{$2$-Star Packing}.

\subsection{The Algorithm}\label{sect:algorithm}
We start with an informal description of the algorithm. 
As stated in the introduction of this section, the central part is Algorithm~\computeCD{}
shown in~\autoref{fig:computeCD}.

\begin{figure}[t]
  \normalsize
\begin{flushleft}
  \alg{\computeCD{}~$(G)$}%
  {An undirected graph~$G$.}%
  {Vertex subsets~$A$ and~$B$ satisfying the three properties of \autoref{thm:bddkernelmain}.}
{
  \begin{pseudocode}
    \item \label{a2:init}    $A := \emptyset, B := \emptyset$
    \item \label{a2:compXY}  Compute a witness~$X$ and the corresponding residual~$Y := V \setminus X$ for~$G$
    \item \label{a2:test}    \textbf{If}~$|Y| \leq (d+1)^2 \cdot |X|$ \textbf{then return}~$(A,B)$
    \item \label{a2:find_ex} $(C,D) \leftarrow$ \findextremal{}~$(G,X,Y)$.
    \item \label{a2:updateG} $G \leftarrow G - (C \cup D); A \leftarrow A \cup C; B \leftarrow B \cup D;$
                                \textbf{goto} line~\ref{a2:compXY}
  \end{pseudocode}
}
\end{flushleft}
  \vspace{-\parsep}
  \caption{Pseudo-code of the main algorithm for computing~$A$ and~$B$.}
  \label{fig:computeCD}
\end{figure}



Using the characterization of bounded-degree graphs by forbidding
large stars, in line~\ref{a2:compXY} \computeCD{} starts with computing two vertex sets~$X$
and~$Y$: First, with a straightforward greedy algorithm, compute
a {\em maximal $(d+1)$-star packing}~ of~$G$, that is,
a set of vertex-disjoint $(d+1)$-stars that cannot be extended by
adding another $(d+1)$-star.  Let~$X$ be the set of vertices of the star
packing. Since the number of stars in the packing is a lower bound for
the size of a minimum bdd-$d$-set,~$X$ is a factor-$(d+2)$ approximate
bdd-$d$-set. Greedily remove vertices from~$X$ such that~$X$ is
still a bdd-$d$-set, and finally set~$Y := V \setminus X$.
We call~$X$ the \emph{witness} and~$Y$ the corresponding \emph{residual}.

\begin{figure}[t]
  \normalsize
\begin{flushleft}
  \proc{\findextremal{}~$(G,X,Y)$}%
  {An undirected graph~$G$, witness~$X$, and residual~$Y$.}%
  {Vertex subsets~$C$ and~$D$ satisfying the local optimality conditions.}%
{
  \begin{pseudocode}
    \item \label{a1:bipartite} $\J{} \leftarrow$ bipartite graph with~$X$ and~$Y$ as its two 
                          vertex subsets and \\
                          $E(\J{})\leftarrow \{\{u,v\} \in E(G) \mid u\in X \text{ and } v\in Y\}$
    \item \label{a2:start}   $F^X_0 \leftarrow \emptyset$ \pccomment{Initialize empty set of forbidden vertices}
    \item \label{a2:oloop}   start with~$j=0$ and \textbf{while}~$F^X_j \not= X$ \textbf{do} \pccomment{Loop while not all vertices in~$X$ are forbidden}
    \item \label{a2:forb}    \pcindent{1} $F^Y_j \leftarrow N_G[N_\J{}(F^X_j)] \setminus X$ \pccomment{Determine forbidden vertices in~$Y$}
    \item \label{a2:pack}    \pcindent{1} $P \leftarrow $ \textbf{star-packing}$(\J{} - (F^X_j \cup F^Y_j), X \setminus F^X_j, Y \setminus F^Y_j,d)$
    \item \label{a2:findHb}  \pcindent{1} $D_0 \leftarrow Y \setminus (F^Y_j \cup V(P))$ \pccomment{Vertices in~$Y$ that are not forbidden and not in~$P$}
    \item \label{a2:line7}                  \pcindent{1} start with~$i=0$ and \textbf{repeat} \pccomment{Start search for~$C,D$ satisfying C\ref{comstr:bddset}} 
    \item \label{a2:findHc} \pcindent{2}     $C_i \leftarrow N_\J{}(D_i)$ 
    \item                   \pcindent{2}     $D_{i+1} \leftarrow N_{P}(C_i) \cup D_i$
    \item                   \pcindent{2}     $i \leftarrow i + 1$
    \item \label{a2:until}  \pcindent{1} \textbf{until}~$D_{i} = D_{i-1}$
    \item \label{a2:findHe} \pcindent{1} $C \leftarrow C_i$,~$D\leftarrow D_i$
    \item \label{a2:cond}    \pcindent{1} \textbf{if}~$C = X \setminus F^X_j$ \textbf{then} \pccomment{$C,D$ also satisfy~C\ref{comstr:D}} 
    \item \label{a2:retker}  \pcindent{2}    \textbf{return}~$(C,D)$ 
    \item \label{a2:setF}    \pcindent{1} $F^X_{j+1} \leftarrow X \setminus C$ \pccomment{Determine forbidden vertices in~$X$ for next iteration}
                                          \hfill 
    \item                    \pcindent{1} $j \leftarrow j + 1$
    \item  \textbf{end while}
    \item  \label{a2:setyy}  $F^Y_j\leftarrow N_G[N_\J{}(F^X_j)]\setminus X$ \pccomment{Recompute forbidden vertices in~$Y$ (as in line~\ref{a2:forb})}
    \item \label{a2:ret} \textbf{return}~$(\emptyset, V\setminus (X\cup F^Y_j))$	
  \end{pseudocode}
}
\vspace*{2ex}
  \textbf{Procedure:} \textbf{star-packing}~$(\J{}, V_1,V_2,d)$\\
  \textbf{Input:} A bipartite graph~$\J{}$ with two vertex subsets~$V_1$
                  and~$V_2$.\\
  \textbf{Output:} A maximum-edge packing of stars that have their
                   centers in~$V_1$ and have at most~$d+1$ leaves
                   in~$V_2$.\\
   See \autoref{lem:starpacking}, the straightforward implementation details using matching techniques are omitted.
\end{flushleft}
  \caption{Pseudo-code of the procedure computing the intermediary vertex subset pair~$(C,D)$.}
  \label{fig:bdd-compression}
\end{figure}
If the residual~$Y$ is too big (condition in line~\ref{a2:test}),
the sets~$X$ and~$Y$ are passed in line~\ref{a2:find_ex} to the
procedure \findextremal{} in~\autoref{fig:bdd-compression}
which 
computes two sets~$C$ and~$D$ satisfying conditions C\ref{comstr:D} and C\ref{comstr:bddset}. 
Computing~$X$ and~$Y$
represents the first step to find a subset pair satisfying 
condition C\ref{comstr:D}:  Since
there is no vertex that has degree more than~$d$ in~$G-X$ (due to the 
fact that~$X$ is a bdd-$d$-set), the search is limited
to those subset pairs where~$C$ is a subset of the witness~$X$
and~$D$ is a subset of~$Y$.


Algorithm~\computeCD{} calls \findextremal{} iteratively until the sets~$A$ and~$B$, 
which are constructed by the union of the outputs of all applications of \findextremal{} (see line~\ref{a2:updateG}), 
satisfy the third property in Theorem~\ref{thm:bddkernelmain}. 
In the following, we intuitively describe
the basic ideas behind \findextremal{}.

To construct the set~$C$ from~$X$, we compute again a star packing~$P$
with the centers of the stars being from~$X$ 
and the leaves being from~$Y$. We relax, on the 
one hand, the requirement that the stars in the packing have exactly~$d+1$ leaves, 
that is, the packing~$P$ might contain $\le d$-stars. 
On the other hand,~$P$ should have a maximum number of edges. 
The rough idea behind the requirement for a maximum number of edges is to
maximize the number of $(d+1)$-stars in~$P$ in the course of the algorithm. Moreover, 
we can observe that, by setting~$C$ equal to the center set of the $(d+1)$-stars 
in~$P$ and~$D$ equal to the leaf set of the $(d+1)$-stars in~$P$,~$C$ is 
a minimum \bdd{} of~$G[C\cup D]$ (condition C\ref{comstr:bddset}). We call 
such a packing a {\em maximum-edge $X$-center $\le(d+1)$-star packing}. 
For computing~$P$, 
the algorithm constructs an auxiliary bipartite graph~$\J{}$ 
with~$X$ as one vertex subset and~$Y$ as the other. 
The edge set of~$\J{}$ consists of the edges in~$G$ with exactly one endpoint 
in~$X$. See line~\ref{a1:bipartite} of~\autoref{fig:bdd-compression}. 
Obviously, a maximum-edge $X$-center 
$\le(d+1)$-star packing of~$G$ corresponds one-to-one with a maximum-edge
packing of stars in~$\J{}$ that have their centers in~$X$ and have at most~$d+1$
leaves in the other vertex subset. Then, the star packing~$P$
can be computed by using techniques for computing maximum matchings 
in~$\J{}$ (in the following, let~\textbf{star-packing}($\J{}$,$V_1$,$V_2$,$d$) 
denote an algorithm that computes a maximum-edge $V_1$-center $\le(d+1)$-star packing~$P$ 
on the bipartite 
graph~$\J{}$).

The most involved part of \findextremal{} in~\autoref{fig:bdd-compression}
is to guarantee that the output subsets in line~\ref{a2:find_ex} fulfill condition
C\ref{comstr:D}. To this end, one uses an iterative approach to compute the
star packing~$P$. Roughly speaking, in each iteration, if the
subsets~$C$ and~$D$ do not fulfill condition C\ref{comstr:D}, then
exclude from further iterations the vertices from~$D$ that themselves or whose neighbors violate
this condition.  See lines~\ref{a2:start}
to~\ref{a2:setF} of~\autoref{fig:bdd-compression} for more details
of the iterative computation.  Herein, for~$j\ge 0$, the sets~$F^X_j
\subseteq X$ and~$F^Y_j \subseteq Y$, where~$F^X_j$ is initialized
with the empty set, and~$F^Y_j$ is computed using~$F^X_j$, store the
vertices excluded from computing~$P$.  To find the vertices that
themselves cause the violation of the condition, that is, vertices
in~$D$ that have neighbors in~$X \setminus C$, one uses an augmenting
path computation in lines~\ref{a2:line7} to~\ref{a2:until} to get in
line~\ref{a2:findHe} subsets~$C$ and~$D$ such that the vertices in~$D$
do \emph{not} themselves violate the condition.  Roughly speaking, the
existence of an edge~$e$ from some vertex in~$D$ to some vertex in~$X
\setminus C$ would imply that the $\le(d+1)$-star packing is not maximum
(witnessed by an augmenting path beginning with~$e$---in principle, this
idea is also used for finding crown decompositions, cf.~\cite{AFLS07}).
The vertices whose neighbors cause the violation of condition
C\ref{comstr:D} are all vertices in~$D$ with neighbors in~$Y \setminus
D$ that themselves have neighbors in~$X \setminus C$.  These neighbors
in~$Y \setminus D$ and the corresponding vertices in~$D$ are excluded
in line~\ref{a2:forb} and line~\ref{a2:setyy}.  We will see that the number of all excluded vertices
is~$O(|X \setminus C|)$, thus, in total, we do not exclude too many vertices
with this iterative method.
The formal proof of correctness is given in the following subsection.

\subsection{Running Time and Correctness}\label{sect:proof}
Now, we show that \computeCD{} in \autoref{fig:computeCD} computes in the claimed time two vertex 
subsets~$A$ and~$B$ that fulfill the three properties given in \autoref{thm:bddkernelmain}. 

\subsubsection{Running Time of \findextremal{}.} 
We begin with the proof of the running time of the procedure
\findextremal{} in~\autoref{fig:bdd-compression}, which uses the
following lemmas.

\begin{mylemma}\label{lem:starpacking}
Procedure {\bf star-packing}$(\J{},V_1,V_2,d)$ in~\autoref{fig:bdd-compression} 
runs in~$O(\sqrt{n} \cdot m)$ time. 
\end{mylemma}

The next lemma is also used for the correctness proof; in particular,
it guarantees the termination of the algorithm. 
\begin{mylemma}\label{lem:F_j}
  If the condition in line~\ref{a2:cond} of~\autoref{fig:bdd-compression} 
is false for a~$j\ge 0$, then~$F^X_j\subsetneq F^X_{j+1}$.
\end{mylemma}
\begin{proof}
In lines~\ref{a2:forb} and~\ref{a2:pack} of~\autoref{fig:bdd-compression}, 
all vertices in~$F^X_j$ 
and their neighbors~$N_\J{}(F^X_j)$ are excluded from the star packing~$P$ in the 
$j$th iteration of the outer loop. Moreover, the vertices in~$N_\J{}(F^X_j)$ 
are excluded from the set~$D_0$
(line~\ref{a2:findHb}). Therefore, a vertex in~$F^X_j$ cannot be added
to~$C$ in line~\ref{a2:findHe}. Thus~$F^X_{j+1}$ (set to~$X \setminus C$ 
in line~\ref{a2:setF})
contains~$F^X_j$. Moreover, this containment is proper, as otherwise the
condition in line~\ref{a2:cond} would be true.
\end{proof}

\begin{mylemma}\label{lem:time}
Procedure \findextremal{} runs in~$O(n^{3/2} \cdot m + n^2)$ time. 
\end{mylemma}

\subsubsection{Correctness of \findextremal{}.} The correctness proof 
for \findextremal{} in~\autoref{fig:bdd-compression} 
is more involved than its running time analysis. 
The following lemmas provide some properties of~$(C,D)$ which are needed.
\begin{mylemma}\label{lem:IHprops}
  For each~$j\ge 0$ the following properties hold after the execution of 
line~\ref{a2:findHe} in~\autoref{fig:bdd-compression}:
  \begin{enumerate}[(1)]
    \setlength{\itemsep}{-\parsep}
    \item \label{IHprops:center} every vertex in~$C$ is a center vertex of a $(d+1)$-star in~$P$, and
    \item \label{IHprops:leaves} the leaves of every star in~$P$ with center in~$C$ are vertices in~$D$.
  \end{enumerate}
\end{mylemma}
\begin{proof}(Sketch)
To prove~(\ref{IHprops:center}), first of all, we show that~$v \in C$ implies~$v \in
        V(P)$, since, otherwise, we could get a $P$-augmenting
        path from some element in~$D_0$ to~$v$. A $P$-augmenting path is a path
        where the edges in~$E(P)$ and the edges not in~$E(P)$ alternate, 
        and the first and the last edge are not in~$E(P)$. 
        This $P$-augmenting path can be constructed in an inductive way by simulating 
        the construction of~$C_i$ in lines~\ref{a2:findHb} to~\ref{a2:until}
        of~\autoref{fig:bdd-compression}. 
        From this $P$-augmenting path, 
        we can then construct a $X$-center $\le(d+1)$-star packing that has 
        more edges than~$P$, contradicting
        that~$E(P)$ has maximum cardinality. Second, every vertex in~$C$ is a center
        of a star due to the definition of~$P$ and Procedure \textbf{star-packing}. Finally, if a
        vertex~$v \in C$ is the center of a star with less than~$(d+1)$
        leaves, then again we get a $P$-augmenting path from some
        element in~$D_0$ to~$v$.

The second statement follows easily from Procedure \textbf{star-packing} and the
pseudo-code in lines~\ref{a2:findHb} to~\ref{a2:findHe}. 
\end{proof}

\begin{mylemma}\label{lem:dist}
  For each~$j\ge 0$ there is no edge in~$G$ between~$D$ and~$N_\J{}(F^X_j)$.
\end{mylemma}
\begin{proof}
The vertices in~$F^X_j$ and the vertices in~$N_G[N_\J{}(F^X_j)]\setminus X$
are excluded from the computation of~$P$ and are not contained in~$D_0$ 
(lines~\ref{a2:forb} to~\ref{a2:findHb} in~\autoref{fig:bdd-compression}). 
Thus, $N_\J{}[F^X_j]\cap D=\emptyset$ 
and therefore there are no edges in~$G$ between~$D$ and~$N_\J{}(F^X_j)$.
\end{proof}

The next lemma shows that the output of~\findextremal{} 
fulfills the local optimality conditions.

\begin{mylemma}\label{lem:correctness}
  Procedure \findextremal{} returns two disjoint vertex subsets
  fulfilling conditions~C\ref{comstr:D} and~C\ref{comstr:bddset}.
\end{mylemma}
\begin{proof}
Clearly, the output consists of two disjoints sets. 
The algorithm returns in lines~\ref{a2:retker} or~\ref{a2:ret} 
of~\autoref{fig:bdd-compression}. 
If it returns in line~\ref{a2:ret}, then the output~$C$ is empty and~$D$ 
contains only
vertices that have a distance at least~$3$ to the vertices in~$X$: 
The condition in line~\ref{a2:oloop} implies~$F^X_j = X$ 
and, therefore,~$F^Y_j$ contains all vertices in~$G \setminus X$ that 
have distance at most~$2$ to the vertices in~$X$. Since~$X$ is a \bdd{} of~$G$, 
all vertices in~$D$ and their neighbors in~$G$ have a degree
at most~$d$. This implies that both conditions hold for the output
returned in this line. It remains to consider the output returned in
line~\ref{a2:retker}.

To show that condition C\ref{comstr:D} holds, recall that~$G - X$  
has maximum degree~$d$ and that~$C \subseteq X$. Therefore, if for a
vertex~$v$ in~$V\setminus X$ we have~$N_\J{}(v) \subseteq C$, then~$v$
has degree at most~$d$ in~$G - C$. Thus, to show that
each vertex in~$N_G[D] \setminus C$ has degree at most~$d$ in~$G - C$, 
it suffices to prove that~$N_\J{}(N_G[D] \setminus C)\subseteq C$. We show 
separately that~$N_\J{}(D) \subseteq C$ and that~$N_\J{}(N_G(D)\setminus C)\subseteq C$.
        
The assignment in line~\ref{a2:findHc} and the until-condition in line~\ref{a2:until}
directly give~$N_\J{}(D) \subseteq C$. Due to \autoref{lem:dist} there
is no edge in~$G$ between~$D$ and~$N_\J{}(F^X_j)$, where~$F^X_j =
X \setminus C$ (the if-condition in line~\ref{a2:cond}, which has to be 
satisfied for the procedure to return in line~\ref{a2:retker}). From this
it follows that the vertices in~$N_G(D) \setminus C$
have no vertex in~$F^X_j$ as neighbor and, thus,~$N_\J{}(N_G(D)\setminus C)\cap 
F^X_j = \emptyset$. Therefore,~$N_\J{}(N_G(D)\setminus C) \subseteq C$.

By Properties~\ref{IHprops:center} and~\ref{IHprops:leaves} of \autoref{lem:IHprops}, 
there are exactly~$|C|$ many vertex-disjoint $(d+1)$-stars in~$G[C\cup D]$. 
Moreover, there is no $(d+1)$-star in~$G[D]$, since~$X$ is a \bdd{} of~$G$. 
Thus, $C$ is a minimum-cardinality \bdd{} of~$G[C\cup D]$. 
\end{proof}


\subsubsection{Running Time and Correctness of \computeCD{}}
To prove the running time and correctness of \computeCD{}, we have to show 
that the output of \findextremal{} contains sufficiently
many vertices of~$Y$. 
To this end, the following lemma plays a decisive role. 
\begin{mylemma}\label{lem:Fsize}
  For all~$j\ge 0$, the set~$F^Y_j$ in line~\ref{a2:forb} and line~\ref{a2:setyy} 
  of~\autoref{fig:bdd-compression}
  has size at most~$(d + 1)^2 \cdot |F^X_j|$.
\end{mylemma}
\begin{proof}
  The proof is by induction on~$j$. The claim trivially holds for~$j=0$, 
  since~$F^Y_0 = \emptyset$. Assume that the claim is true for~$j > 0$.
  Since~$F^X_{j} \subsetneq F^X_{j+1}$ (\autoref{lem:F_j}), we have
  \[F^Y_{j+1} = F^Y_j \cup N_{G-X}[N_{\J{} - F^Y_j}(F^X_{j+1}\setminus F^X_j)]. \]
  We first bound the size of~$N_{\J{} - F^Y_j}(F^X_{j+1} \setminus F^X_j)$.  
  Since~$F^X_{j+1}$ was set to~$X \setminus C$ at the end of the $j$th iteration 
  of the outer loop (line~\ref{a2:setF}), the vertices in~$N_{\J{} - F^Y_j}(F^X_{j+1} \setminus F^X_j)$ were
  not excluded from computing the packing~$P$ (line~\ref{a2:pack}) of the $j$th iteration. 
  Moreover,~$N_{\J{}-F^Y_j}(F^X_{j+1} \setminus F^X_j) \subseteq
  V(P)$ for the star packing~$P$ computed in the $j$th iteration, since, 
  otherwise, the set~$D_0$ in line~\ref{a2:findHb} would contain a vertex~$v$
  in~$N_{\J{}-F^Y_j}(F^X_{j+1} \setminus F^X_j)$ and, then, line~\ref{a2:findHc}
  would include~$N_\J{}(v)$ into~$C$, which would contradict 
  the fact that~$C \cap F^X_{j+1} = \emptyset$ (line~\ref{a2:setF}). Due to
  property~\ref{IHprops:leaves} in \autoref{lem:IHprops} the leaves of
  every star in~$P$ with center in~$C$ are vertices in~$D$ and, thus, the vertices
  in~$N_{\J{}- F^Y_j}(F^X_{j+1} \setminus F^X_j)$ are leaves
  of stars in~$P$ with centers in~$F^X_{j+1} \setminus F^X_j$. Since each
  star has at most~$(d+1)$ leaves, the set~$N_{\J{}-F^Y_j}(F^X_{j+1}
  \setminus F^X_j)$ has size at most~$(d+1) \cdot |{F^X_{j+1} \setminus F^X_j}|$.
  The remaining part is easy to bound: since all the vertices in~$V\setminus X$ have 
  degree at most~$d$, we get
  \begin{align*}
    \left|N_{G-X}[N_{\J{}-F^Y_j}
     (F^X_{j+1} \setminus F^X_j)]\right| & \leq (d \cdot (d+1) + (d+1)) 
                                 \cdot |{F^X_{j+1} \setminus F^X_j}|\\ 
                               & = (d+1)^2 \cdot |{F^X_{j+1} \setminus F^X_j}|.
  \end{align*}
  With the induction hypothesis, we get that
 \begin{align*}
    |F^Y_{j+1}|
       & \leq  |F^Y_j| + |N_{G-X}[N_{\J{}-F^Y_j}(F^X_{j+1} \setminus
            F^X_j)]|\\
       & =  (d+1)^2  \cdot |F^X_j| + (d+1)^2 \cdot |{F^X_{j+1} \setminus F^X_j}|
        =  (d+1)^2  \cdot |F^X_{j+1}|.
  \end{align*}
\end{proof}

\begin{mylemma}\label{lem:sizebound}
  Procedure \findextremal{} always finds two sets~$C$ and~$D$ such
  that~$|Y \setminus D| \leq (d+1)^2 \cdot |X \setminus C|$.
\end{mylemma}
\begin{proof}
  If \findextremal{} terminates, then~$V'=F^X_j\cup
  F^Y_j$ for the graph~$G'=(V',E')$ resulting by removing~$C\cup
  D$ from~$G$.  Since~$C \subseteq X$ and~$D \subseteq Y$, we
  have~$X \setminus C = F^X_j$ and~$Y \setminus D = F^Y_j$, and by
  \autoref{lem:Fsize} it follows immediately that~$|Y \setminus D| \leq
  (d+1)^2 \cdot |X \setminus C|$.
\end{proof}
Therefore, if~$|Y| > (d+1)^2 \cdot |X|$,
then \findextremal{} always returns two sets~$C$ and~$D$ such that~$D$
is not empty. 

\begin{mylemma}\label{lem:computecd-time}
Algorithm~\computeCD{} runs in~$O(n^{5/2} \cdot m + n^3)$ time. 
\end{mylemma}


\begin{mylemma}\label{lem:sizecondition}
  The sets~$A$ and~$B$ computed by \computeCD{} fulfill the three properties
  given in \autoref{thm:bddkernelmain}.
\end{mylemma}
\begin{proof}
Since every~$(C,D)$ output by \findextremal{} in line~\ref{a2:find_ex} of 
\computeCD{} in \autoref{fig:computeCD} fulfills conditions~C\ref{comstr:D}
and~C\ref{comstr:bddset}
(\autoref{lem:correctness}),
the pair~$(A,B)$ output in line~\ref{a2:test} of \computeCD{} fulfills
conditions~C\ref{comstr:D}
and~C\ref{comstr:bddset}, and, therefore, also
the local optimality conditions (\autoref{prop:commitmentreduction}).
It remains to show that~$(A,B)$
fulfills the size condition.

  Let~$X$ and~$Y$ be the last computed witness and residual, respectively.
  Since the condition in line~\ref{a2:test} is true, we know that~$|Y|
  \leq (d+1)^2 \cdot |X|$.  Recall that~$X$ is a factor-$(d+2)$
  approximate bdd-$d$-set for~$G' := G - (A \cup B)$. Thus, every \bdd{}
  of~$G'$ has size at least~$|X|/(d+2)$. Since the output sets~$A$ and~$B$
  fulfill the local optimality conditions and the bounded-degree property
  is hereditary, every \bdd{} of~$G'$ has size at least
\begin{align*}
\frac{|X|}{d+2} \stackrel{(*)}{\ge}  \frac{|V'|}{(d+2)((d+1)^2+1)}
 = \frac{|V'|}{(d^3+4d^2+6d+4)}.
\end{align*}
The inequality~(*) follows from the fact that~$Y$ is small, that is,~$|Y|
\leq (d+1)^2 \cdot |X|$ (note that~$V' = X \cup Y$).
\end{proof}
With Lemmas~\ref{lem:computecd-time} and~\ref{lem:sizecondition}, the proof of \autoref{thm:bddkernelmain} is completed.
\section{Conclusion}
Our main result is to generalize the Nemhauser-Trotter-Theorem, which
applies to the \BDD{} problem with $d=0$ (that is, \textsc{Vertex Cover}),
to the general case with arbitrary $d\geq 0$. In particular, 
in this way we contribute problem kernels with a number of vertices
linear in the solution size~$k$ for all constant 
values of~$d$ for \BDD{}.
To this end, we developed a new algorithmic strategy that is based on
extremal combinatorial arguments. The original NT-Theorem~\cite{NT75}
has been proven using linear programming relaxations---we see no way how
this could have been generalized to \BDD{}. By way of contrast, we
presented a purely combinatorial data reduction algorithm which is also
completely different from known combinatorial data reduction algorithms
for \textsc{Vertex Cover} (see~\cite{AFLS07,BE85,CC08}). 
Finally, Baldwin et al.~\cite[page~175]{BCK+05} remarked that, with
respect to 
practical applicability in the case of \textsc{Vertex Cover} kernelization, 
combinatorial data reduction algorithms are more powerful than ``slower
methods that rely on linear programming relaxation''.
Hence, we expect that benefits similar to those derived from 
\textsc{Vertex Cover} kernelization for biological network analysis 
(see the motivation part of our introductory discussion) may be provided by  
\BDD{} kernelization.
\bibliographystyle{abbrv}
\bibliography{fellows}

\newcommand{\bibremark}[1]{\marginpar{\tiny\bf#1}}
\begin{thebibliography}{10}

\bibitem{AFLS07}
F.~N. Abu{-}Khzam, M.~R. Fellows, M.~A. Langston, and W.~H. Suters.
\newblock Crown structures for vertex cover kernelization.
\newblock {\em Theory Comput.\ Syst.}, 41(3):411--430, 2007.

\bibitem{BBHS06}
B.~Balasundaram, S.~Butenko, I.~V. Hicks, and S.~Sachdeva.
\newblock Clique relaxations in social network analysis: The maximum $k$-plex
  problem.
\newblock Manuscript, 2008.

\bibitem{BCK+05}
N.~Baldwin, E.~Chesler, S.~Kirov, M.~Langston, J.~Snoddy, R.~Williams, and
  B.~Zhang.
\newblock Computational, integrative, and comparative methods for the
  elucidation of genetic coexpression networks.
\newblock {\em Journal of Biomedicine and Biotechnology}, 2(2005):172--180,
  2005.

\bibitem{BE85}
R.~Bar-Yehuda and S.~Even.
\newblock A local-ratio theorem for approximating the weighted vertex cover
  problem.
\newblock {\em Ann.\ of Discrete Math.}, 25:27--45, 1985.

\bibitem{BP08}
H.~L. Bodlaender and E.~Penninkx.
\newblock A linear kernel for planar feedback vertex set.
\newblock In {\em Proc.\ 3rd IWPEC}, volume 5018 of {\em LNCS}, pages 160--171.
  Springer, 2008.

\bibitem{BF01}
H.~L. Bodlaender and B.~van Antwerpen-de Fluiter.
\newblock Reduction algorithms for graphs of small treewidth.
\newblock {\em Inform.\ and Comput.}, 167(2):86--119, 2001.

\bibitem{CKJ01}
J.~Chen, I.~A. Kanj, and W.~Jia.
\newblock Vertex cover: Further observations and further improvements.
\newblock {\em J.~Algorithms}, 41(2):280--301, 2001.

\bibitem{CLS+05}
E.~J. Chesler, L.~Lu, S.~Shou, Y.~Qu, J.~Gu, J.~Wang, H.~C. Hsu, J.~D. Mountz,
  N.~E. Baldwin, M.~A. Langston, D.~W. Threadgill, K.~F. Manly, and R.~W.
  Williams.
\newblock Complex trait analysis of gene expression uncovers polygenic and
  pleiotropic networks that modulate nervous system function.
\newblock {\em Nature Genetics}, 37(3):233--242, 2005.

\bibitem{CC08}
M.~Chleb\'{i}k and J.~Chleb\'{i}kov\'{a}.
\newblock Crown reductions for the minimum weighted vertex cover problem.
\newblock {\em Discrete Appl.\ Math.}, 156:292--312, 2008.

\bibitem{DF99}
R.~G. Downey and M.~R. Fellows.
\newblock {\em Parameterized Complexity}.
\newblock Springer, 1999.

\bibitem{FG06}
J.~Flum and M.~Grohe.
\newblock {\em Parameterized Complexity Theory}.
\newblock Springer, 2006.

\bibitem{Guo08}
J.~Guo.
\newblock A more effective linear kernelization for cluster editing.
\newblock {\em Theor.\ Comput.\ Sci.}, 2008.
\newblock To appear.

\bibitem{GN07SIGACT}
J.~Guo and R.~Niedermeier.
\newblock Invitation to data reduction and problem kernelization.
\newblock {\em ACM SIGACT News}, 38(1):31--45, 2007.

\bibitem{GN07ICALP}
J.~Guo and R.~Niedermeier.
\newblock Linear problem kernels for {NP}-hard problems on planar graphs.
\newblock In {\em Proc.\ 34th ICALP}, volume 4596 of {\em LNCS}, pages
  375--386. Springer, 2007.

\bibitem{KPXS09}
I.~A. Kanj, M.~J. Pelsmajer, G.~Xia, and M.~Schaefer.
\newblock On the induced matching problem.
\newblock {\em J.~Comput.\ System Sci.}, 2009.
\newblock To appear.

\bibitem{KR08}
S.~Khot and O.~Regev.
\newblock Vertex cover might be hard to approximate to within $2-\epsilon$.
\newblock {\em J.~Comput.\ System Sci.}, 74(3):335--349, 2008.

\bibitem{Khu02}
S.~Khuller.
\newblock The {Vertex Cover} problem.
\newblock {\em ACM SIGACT News}, 33(2):31--33, 2002.

\bibitem{KHMN07}
C.~Komusiewicz, F.~H{\"u}ffner, H.~Moser, and R.~Niedermeier.
\newblock Isolation concepts for enumerating dense subgraphs.
\newblock In {\em Proc.\ 13th COCOON}, volume 4598 of {\em LNCS}, pages
  140--150. Springer, 2007.

\bibitem{Lan08pers}
M.~A. Langston, 2008.
\newblock Personal communication.

\bibitem{NT75}
G.~L. Nemhauser and L.~E. Trotter.
\newblock Vertex packings: Structural properties and algorithms.
\newblock {\em Math.\ Program.}, 8:232--248, 1975.

\bibitem{Nie06}
R.~Niedermeier.
\newblock {\em Invitation to Fixed-Parameter Algorithms}.
\newblock Oxford University Press, 2006.

\bibitem{NRT05}
N.~Nishimura, P.~Ragde, and D.~M. Thilikos.
\newblock Fast fixed-parameter tractable algorithms for nontrivial
  generalizations of {V}ertex {C}over.
\newblock {\em Discrete Appl.\ Math.}, 152(1--3):229--245, 2005.

\bibitem{PS06}
E.~Prieto and C.~Sloper.
\newblock Looking at the stars.
\newblock {\em Theor.\ Comput.\ Sci.}, 351(3):437--445, 2006.

\bibitem{SF78}
S.~B. Seidman and B.~L. Foster.
\newblock A graph-theoretic generalization of the clique concept.
\newblock {\em Journal of Mathematical Sociology}, 6:139--154, 1978.

\bibitem{WNFC08}
J.~Wang, D.~Ning, Q.~Feng, and J.~Chen.
\newblock An improved parameterized algorithm for a generalized matching
  problem.
\newblock In {\em Proc.\ 5th TAMC}, volume 4978 of {\em LNCS}, pages 212--222.
  Springer, 2008.

\end{thebibliography}

\end{document}